\documentclass[12pt]{article}

\usepackage{amsmath,amsfonts,amssymb,latexsym,amsthm,verbatim,a4wide,graphicx,hyperref}

 \theoremstyle{plain}
\newtheorem{theorem}{Theorem}[section]
\newtheorem{lemma}[theorem]{Lemma}
\newtheorem{proposition}[theorem]{Proposition}
\newtheorem{corollary}[theorem]{Corollary}

\newtheorem{definition}[theorem]{Definition}
\newtheorem{example}[theorem]{Example}

\newtheorem{remark}[theorem]{Remark}

\newcommand{\pr}{{\mathbb P}}
\newcommand{\vc}[1]{{\mathbf{#1}}}
\newcommand{\vd}[1]{{\boldsymbol{#1}}}
\newcommand{\suc}{{\rm suc}}
\newcommand{\II}{{\mathbb I}}
\newcommand{\sett}[1]{{\mathcal{#1}}}

\newcommand{\typv}[1]{{\mathcal{T}}_{#1}}
\newcommand{\typ}{\typv{\beta}}
\newcommand{\intyp}{{\mathcal{I}}_{\beta}}
\newcommand{\typsmb}{\typv{\epsilon}^{(i)}}
\newcommand{\curP}[1]{\Pi_{\vc{U}}(#1)}
\newcommand{\curPi}[1]{\Pi_{\vc{U}^{(i)}}(#1)}
\newcommand{\mat}[1]{\mathsf{#1}}
\newcommand{\bin}{{\rm Bin}}

\newcommand{\elnt}{L^{*}_{N,T}}
\newcommand{\elrt}{L^{*}_{R,T}}
\newcommand{\tends}{\rightarrow \infty}
\newcommand{\whu}{\widehat{\vc{U}}}
\newcommand{\prV}[1]{P_{\vc{#1}}}
\newcommand{\qrV}[1]{Q_{\vc{#1}}}
\newcommand{\prv}[2]{P_{\vc{#1}}\left( \vc{#2} \right)}
\newcommand{\qrv}[2]{Q_{\vc{#1}}\left( \vc{#2} \right)}
\newcommand{\qst}[1]{Q^*\left( \vc{#1} \right)}
\newcommand{\pmax}[1]{p_{\max}(\vc{#1})}
\newcommand{\vtheta}{\vd{\theta}}
\newcommand{\pcaus}{P_{\mat{X} \| \vc{Y}^-} \left( \sett{X} \| \vc{y}^- \right)}
\newcommand{\pcausgen}{P_{\vc{X} \| \vc{Y}}\left( \vc{x} \| \vc{y} \right)}

\begin{document}

\title{Strong converses for group testing in the finite blocklength regime}
\date{\today}
\author{Oliver Johnson\thanks{School of Mathematics, Univeristy of Bristol, University Walk, Bristol, BS8 1TW, UK. Email: {\tt maotj@bristol.ac.uk}}}
\maketitle

\begin{abstract} \noindent
We prove new strong converse results in  a variety of group testing settings, generalizing a result of Baldassini, Johnson and Aldridge. These results are proved by two distinct approaches,
corresponding to the non-adaptive and adaptive cases. In the non-adaptive case, we mimic the hypothesis testing 
argument introduced in the finite blocklength channel
coding
regime by Polyanskiy, Poor and Verd\'{u}. In the adaptive case, we combine a formulation based on directed information
 theory with
ideas of Kemperman, Kesten and Wolfowitz from the problem of channel coding with feedback.
In both cases, we prove results which are valid for finite sized problems, and imply capacity results in the asymptotic
 regime.
These results are illustrated graphically for a range of models.
\end{abstract}

\section{Introduction and group testing model}

The group testing problem was introduced by Dorfman \cite{dorfman} in the 1940s, and captures the idea of 
efficiently isolating
a small subset  $\sett{K}$ of defective items  in a larger set containing $N$ items.  
The models used vary slightly, but the fundamental setting is that we 
perform a sequence of tests, each defined by a testing pool of items. The outcome of the test depends on 
the number of defective items in the pool.
The most basic model, which we refer to as `standard noiseless group testing' is that the test outcome 
equals 1 if and only if the testing pool contains at least one defective item.  Given $T$ tests,
the group testing problem requires us to design test pools and estimation algorithms   to maximise  $\pr(\suc)$,
the success probability (probability of 
recovering the defective set exactly).

This paper focuses on converse results, giving upper bounds on the $\pr(\suc)$ that can be 
 achieved by any algorithm given $T$ tests. We generalize the following
 strong result proved by Baldassini, Johnson and Aldridge \cite[Theorem 3.1]{johnsonc10}:
\begin{theorem} \label{thm:bja}
 Suppose  the defective set $\sett{K}$ is chosen
uniformly from the $\binom{N}{K}$ possible sets of given  size $K$. For adaptive or non-adaptive  standard noiseless group testing:
\begin{equation} \label{eq:johnson}
 \pr(\suc) \leq \frac{2^T}{\binom{N}{K}}.
\end{equation}
\end{theorem}
We extend this result to a  variety of settings.
 We first discuss four dichotomies in the modelling of the group testing
problem; (i) combinatorial vs probabilistic  (ii) binary vs non-binary (iii) noisy vs noiseless (iv) adaptive vs
non-adaptive. There are a number of further variations beyond these, as described in an ever-increasing body of literature.

\begin{enumerate}
\item { \bf [Combinatorial vs Probabilistic]}
The first  categorisation concerns the way  the defective items are chosen. 
Combinatorial group testing (see for example \cite{johnson33,atia,johnsonc10,chan}) is the model from Theorem \ref{thm:bja}: we suppose there is a fixed number $K$ of defective items, and the defective set $\sett{K}$ is chosen
uniformly from the $\binom{N}{K}$ possible sets of this size. In probabilistic group testing (see for example \cite{johnsonc12,li5})
the $i$th item is defective independently with probability $p_i$ (with $p_i$ not necessarily identical). In fact, we put both
these models in a common setting:

\begin{definition} We write $\vc{U} \in \{ 0, 1 \}^N$ for the (random) defectivity vector, where component $U_i$ is the indicator 
of the event that the $i$th item is defective.
For any vector $\vc{u} \in \{0,1 \}^N$ write
$ \prv{U}{u}  = \pr( \vc{U} = \vc{u})$, and define entropy  
\begin{equation} \label{eq:entropydef} H(\vc{U}) = -\sum_{\vc{u} \in \{0,1 \}^N} \prv{U}{u} \log_2 \prv{U}{u}.\end{equation}
\end{definition}
\begin{example}
For the two models as described above:
\begin{enumerate}
\item For combinatorial group testing, since $\vc{U}$ is uniform over $\binom{N}{K}$ outcomes, the 
entropy $H(\vc{U}) = \log_2 \binom{N}{K}$.
\item For probabilistic group testing, the entropy $H(\vc{U}) = \sum_{i=1}^N h(p_i)$, where $h(t) := - t \log_2 t
- (1-t) \log_2 (1-t)$ is the binary entropy function. If $p_i \equiv p$ then $H(\vc{U}) = N h(p)$.
\end{enumerate}
\end{example}
However, Corollary \ref{cor:inforate} below shows that results resembling Theorem \ref{thm:bja} can be proved for general
sources satisfying the Shannon-McMillan-Breiman theorem. This includes settings where the defectivity vector is generated
by a stationary ergodic Markov chain, which is a natural model of a setting where nearest neighbours are susceptible to
infection.

\item {\bf [Binary vs Non-binary]}
The second variation comes  in terms of the set of outcomes $\sett{Y}$ that each test can produce.
 We
refer to $\sett{Y}$ as the alphabet,  since in this paper (as in \cite{johnson33,atia} and other papers) we 
consider an analogy between  group testing and the
channel coding problem. It is most standard to consider the binary case, where $\sett{Y} = \{ 0, 1 \}$, though other models are 
possible (see for example \cite[Section 6.3]{aldridge} for a detailed review). For brevity we will only consider the binary case in this paper (though our techniques will be valid in a more general setting). We write $\vc{Y} \in \sett{Y}^T = \{ 0,1 \}^T$ for the outcome of the group testing process. 

\item { \bf[Noisy vs Noiseless]}
The third difference concerns the way in which the outcome of each test is formed. To fix notation,
we perform a sequence of $T$ tests defined by test pools $\sett{X}_1, \ldots, \sett{X}_T$, where each $\sett{X}_t \subseteq \{ 1, 2, \ldots, N \}$.  We 
represent this by a binary test matrix 
$\sett{X} = (x_{it} : i = 1, \ldots, N \mbox{\;and\;} t = 1,\ldots, T)$,
where $x_{it} = 1$ if and only if item $i$ is included in the $t$th pool (a concatenation of column vectors 
 given by the indicator functions of the $T$ test pools). Since the test design may be random, we write $\mat{X}$ for a random
variable giving a test matrix of this form.

For a test matrix $\mat{X}$ and defectivity vector $\vc{u}$, a key object of interest is the vector $\vc{K} = \vc{u}^T \mat{X}$.
Here, the $t$th component of $\vc{K}$ is
$K_t = \sum_{i=1}^N x_{it} \II( \mbox{item $i \in \mat{X}$})$, the total number of defective items appearing in the $t$th test. 
Observe that $K_t$ is a deterministic function of $\vc{u}$ and $\mat{X}_t$ (and does not depend on any other
variables). It is useful to define $\vc{X}$ via $X_t  = \II(K_t \geq 1)$.

We assume that the group testing model
 is static, memoryless and satisfies
 the `Only Defects Matter' property introduced
by Aldridge \cite{aldridge2, aldridge}:
\begin{definition}[Only Defects Matter] \label{def:odm}
We assume that the $t$th test outcome $Y_t$ is a random function of $K_t$ (so $\vc{Y}$ is conditionally independent of 
$\vc{U}$ given $\vc{K}$, and $Y_t$ is conditionally independent of $(K_s)_{s \neq t}$ given $K_t$).
Further, for some fixed transition matrix $P$, we assume
\begin{equation} \label{eq:odm}
  \pr( Y_t = y|  K_t = k) = P(y | k), \mbox{\;\;\; for all $y,k,t$}. \end{equation}
\end{definition}
Note that Definition \ref{def:odm} includes the noiseless standard group testing case, where we simply take $\vc{Y} = \vc{X}$.
To  understand Definition \ref{def:odm}, we can
 consider the case where $\vc{X}$ is fed (symbol by symbol)
through a memoryless noisy channel,
 independent of the defectivity vector $\vc{U}$, and any randomness in the group testing design. In the notation of 
\eqref{eq:odm}, we assume that $P(1 |k )  \equiv P(1|1)$ for all $k \geq 1$; in the noiseless case we take
$P(1 | k) \equiv 1$ and $P(0 | 0) = 1$.
However, Definition \ref{def:odm} allows a 
wider range of noise models, including the dilution channel of Atia and Saligrama \cite{atia}, where we take 
$P(0|k) = (1-u)^k$ for some $u$.

For a fixed test matrix $\mat{X} = \sett{X}$, as in \cite{johnsonc10}, in the noiseless case the testing procedure naturally 
defines a mapping $\vtheta( \cdot, \sett{X}): \{0,1 \}^N \rightarrow \{ 0, 1 \}^T$.
That is, given defectivity vector $\vc{u} \in \{ 0,1 \}^N$, we write the vector function $\vtheta$ with components given by scalar
function $\theta$, in the form
\begin{equation} \label{eq:vecfn}
\vtheta(\vc{u}, \sett{X}) = \left( \theta(\vc{u}, \sett{X}_1), \theta(\vc{u}, \sett{X}_2), \ldots ,
, \theta(\vc{u}, \sett{X}_T) \right). \end{equation}
where $\theta( \vc{u}, \sett{X}_t) =
 \II(K_t \geq 1) = X_t$.
\item {\bf [Adaptive vs Non-adaptive]}
The final distinction  is whether we design the test matrix using an adaptive or a non-adaptive strategy.
In the non-adaptive case the entire test matrix $\mat{X} = \sett{X}$ needs to be chosen in advance of 
the tests. In contrast, in the adaptive case, the  $(t+1)$st test pool $\sett{X}_{t+1}$ is chosen based on a knowledge of previous 
test pools $\sett{X}_{1,t} := \left\{  \sett{X}_1, \ldots, \sett{X}_{t} \right\}$ and test outcomes $
\vc{Y}_{1,t} := \left\{ Y_1, \ldots, Y_{t} \right\}$.
We  can think (see \cite{aldridge})
that adaptive group testing corresponds to channel coding with feedback, and non-adaptive group
testing to coding with no feedback.
Clearly
(see \cite{johnsonc10}), 
we can do no worse in the adaptive setting than for non-adaptive group testing, but it remains an open and interesting question to determine precisely by how much adaptivity can improve performance.
 
We argue that
a key tool in understanding adaptive group testing is directed information theory. This was first introduced by Marko \cite{marko} in the 
1970s, with interest revived by the work of Massey \cite{massey} in the 1990s, and developed further by 
authors such as Kramer \cite{kramer4} and Schreiber \cite{schreiber}. In particular, as described by Massey \cite{massey},
many authors make an incorrect probabilistic formulation
of such simple objects as discrete memoryless channels with feedback. A correct formulation requires the use of  the causal conditional probability distribution as studied for example by Kramer \cite{kramer4}. We use the notation
of the review paper \cite[Equation (7)]{amblard}, that for sequences $\vc{x} = (x_1, \ldots, x_T)$ and 
$\vc{y} = (y_1, \ldots, y_T)$, and subsequences $\vc{y}_{1,t-1} = (y_1, \ldots, y_{t-1})$,
\begin{equation} \label{eq:causalcond}
\pcausgen := \prod_{t=1}^T \prv{X_t|X_{1,t-1}, Y_{1,t}}{x_t | x_{1,t-1}, y_{1,t}}.
\end{equation}
Note that for any fixed $\vc{y}$, the fact that it is formed as a product of probability distributions means that $\sum_{\vc{x}} \pcausgen = 1$.
Using this probability distribution implies the form of the
directed information of Marko \cite{marko} (see also the later definition of transfer entropy by Schreiber \cite{schreiber}).

In Lemma \ref{lem:jointprob} below, assuming the Only Defects Matter property Definition \ref{def:odm},
we decompose the joint probability of $(\vc{U}, \sett{X}, \vc{Y})$ in the 
general adaptive setting, using  the term  $\pcaus$, which is defined in
\eqref{eq:jointprob2}. Here we use the causal conditional probability  notation \eqref{eq:causalcond} above,
with superscript $\vc{y}^{-}$ referring to the fact that there is a lag in the index of $\vc{y}$
in \eqref{eq:jointprob2} compared with \eqref{eq:causalcond} (we choose the set $\sett{X}_i$ based on a knowledge of the
previous sets $\sett{X}_{1,t-1}$ and test outcomes $\vc{y}_{1,t-1}$). 

The decomposition in Lemma \ref{lem:jointprob} shows that  the $t$th output symbol $Y_t$ is conditionally
independent of $\vc{U}$, given values $K_{1,t}$ and previous outputs $\vc{Y}_{1,t-1}$. This is precisely the
definition of a causal system between $\vc{K}$ and $\vc{Y}$
 given by Massey in \cite[Equation (8)]{massey}, under which condition the capacity of a discrete
memoryless channel is not increased by feedback. 
\end{enumerate}
Regardless of these variations, we always make an estimate $\vc{Z} = \whu$, based only on a knowledge of
outputs $\vc{Y} = \vc{y}$ and  test matrix $\mat{X} = \sett{X}$, using
 a probabilistic estimator (decoder) that gives $\vc{Z} = \vc{z}$ with
probability  \begin{equation} \label{eq:decoder} \prv{Z|Y,\mat{X}}{z|y,\sett{X}}. \end{equation}

The main results of the paper are Theorem \ref{thm:ppvmain}, which gives an upper bound on $\pr(\suc)$ in the non-adaptive 
case, and Proposition \ref{prop:basic}, which gives the corresponding result in the adaptive case.
The strength of these results is illustrated in results such as 
Examples \ref{ex:bsc} and \ref{ex:adaptivenoisy}, where we calculate
bounds on the success probability in the case where $\vc{X}$ forms the input and $\vc{Y}$ the output of a binary symmetric
channel with error probability $p$.
 We illustrate these bounds in Figure \ref{fig:bscadaptive}, where the upper bounds on $\pr(\suc)$ in both adaptive and
non-adaptive cases are plotted in the finite blocklength case of combinatorial group testing with $N=500$, $K = 10$, $p=0.11$.

The structure of the paper is as follows.
 In Section \ref{sec:existing} we review existing results concerning 
group testing converses.
In Section \ref{sec:hyptest} we use an argument based on the paper by Polyanskiy, Poor and Verd\'{u} \cite{polyanskiy2} to
prove Theorem \ref{thm:ppvmain}, giving
 a strong converse for non-adaptive group testing, and discuss the bounds for the  binary symmetric channel
 case. In Section \ref{sec:adaptive} we discuss the adaptive case, by extending arguments first given
for channels with feedback by Kemperman \cite{kemperman}, Kesten \cite{kesten}
 and Wolfowitz \cite{wolfowitz3} (see also Gallager \cite{gallager2}).
We prove a bound (Proposition \ref{prop:basic}) 
which specializes 
 in the noiseless case to give a result (Theorem \ref{thm:noiseless}) which generalizes
Theorem \ref{thm:bja}. We consider examples of this noiseless result in the  probabilistic case in Section \ref{sec:noiselessexamples}.
 Finally in Section \ref{sec:adaptivenoisy} we apply  Proposition \ref{prop:basic} in the noisy adaptive case.
The proofs of the main theorems  are given in Appendices.

While this paper only considers group testing, we remark that this problem lies in the area of sparse inference,
which includes problems such as compressed sensing and matrix completion, 
and it is likely that results proved here will extend to  more
general settings. The paper \cite{aksoylar} gives a review of links between group testing and other sparse inference problems.
Group testing itself has a number of applications, including cognitive radios \cite{aldridge2,atia2,johnsonc12}, network
tomography \cite{cheraghchi} and efficient
gene sequencing \cite{erlich,shental}. The bounds proved here should provide fundamental performance limits in these
contexts.
\section{Existing converse results} \label{sec:existing}
 It is
clear from information-theoretic considerations that to find all the defectives in the noiseless case  will require  at least
 $T^* = H(\vc{U})$ (the ``magic number'')
tests.  In the language of channel coding, the focus of this paper is on converse results; that is
 given $O(T^*)$ tests, we give strong upper bounds  on the success probability
$\pr(\suc)$ of any possible algorithm.  

There has been considerable work
on the achievability part of the problem, in developing group testing algorithms and proving performance guarantees.
 Early work on group testing considered algorithms which could be proved to be order optimal (see for example the analysis of \cite{balding,du,dyachkov3}), often using combinatorial properties such
as separability or disjunctness. 
More recently there has been interest
(see for example \cite{johnson33,aldridge3,atia,chan,li5,malyutov,malyutov4,scarlett,wadayama,wadayama2})
in finding the best possible constant, that is to find algorithms which succeed 
with high probability using $T = c T^* = c H(\vc{U})$ tests, for $c$ as small as possible.
In this context, the paper \cite{johnsonc10} 
defined the capacity of combinatorial group testing problems, a definition  extended to both combinatorial and
probabilistic group testing in \cite{johnsonc12}. We state this definition for both weak and strong capacity in the sense of Wolfowitz: 
\begin{definition} \label{def:capacity} Consider a sequence of group testing problems where the $i$th problem
has defectivity vector  $\vc{U}^{(i)}$, and consider algorithms which are given $T(i)$ tests. We think of $H(\vc{U}^{(i)})/T(i)$ (the number of bits of information learned per test) as the rate of the algorithm and refer to a constant $C$ as the weak group testing capacity if for any $\epsilon > 0$:
  \begin{enumerate}
    \item any sequence of algorithms with
      \begin{equation} \label{eq:lower}
        \liminf_{i \tends} \frac{ H(\vc{U}^{(i)}) }{T(i)} \geq C+ \epsilon,
      \end{equation}
      has success probability $\pr(\suc)$ bounded away from 1,
    \item and there exists a sequence of algorithms with
      \begin{equation} \label{eq:upper}
        \liminf_{i \tends} \frac{H(\vc{U}^{(i)}) }{T(i)}  \geq C - \epsilon
      \end{equation}
      with success probability $\pr(\suc) \rightarrow 1$.
  \end{enumerate}
We call $C$  the strong capacity if $\pr(\suc) \rightarrow 0$ for any sequence of algorithms satisfying \eqref{eq:lower}.
\end{definition}
For example, in \cite{johnsonc10} we prove that noiseless adaptive combinatorial group testing has strong capacity 1. This result is proved by combining Hwang's Generalized
Binary Splitting Algorithm \cite{hwang} (which is essentially optimal -- see also \cite{johnsonc10,du} for a discussion of this)
with the converse result Theorem \ref{thm:bja}. However, even in the noiseless non-adaptive case
 the optimal algorithm remains unclear,  although some results are known in some
regimes, under assumptions about the distribution of $\mat{X}$
 (see for example \cite{aldridge3,johnson33,scarlett,wadayama}). 

In general,
capacity results are asymptotic in character, whereas we will consider the finite blocklength regime (in the spirit of \cite{polyanskiy2})
and prove bounds on $\pr(\suc)$ for any size of 
problem.
We briefly review existing converse results. First, we mention that results (often referred to as folklore) can be proved using
arguments based on Fano's inequality. 
\begin{lemma}  Using $T$ tests:
\begin{enumerate}
\item For combinatorial group testing Chan et al. \cite[Theorem 1]{chan} give 
\begin{equation} \label{eq:chan}
 \pr(\suc) \leq \frac{T}{\log_2 \binom{N}{K}}.
\end{equation} 
\item For probabilistic group testing Li et al. \cite[Theorem 1]{li5} give
\begin{equation} \label{eq:li}
 \pr(\suc) \leq \frac{T}{N h(p)}.
\end{equation} 
\end{enumerate}
\end{lemma}
In order to understand the relationship between \eqref{eq:chan} and Theorem \ref{thm:bja}; fix $\delta > 0$ and
use  $T = T^*(1-\delta)$ tests, in a regime where
$\log_2 \binom{N}{K} \rightarrow \infty$ and hence $T^* \rightarrow \infty$. Chan et al.'s result \eqref{eq:chan}
gives that $\pr(\suc) \leq (1-\delta)$, whereas \eqref{eq:johnson} implies $\pr(\suc) \leq 2^{-\delta T^*}$.  In the language of Definition \ref{def:capacity},
 Chan et al. \cite{chan} give a weak converse  whereas Baldassini, Johnson and Aldridge \cite{johnsonc10}
give a strong converse. In fact, \eqref{eq:johnson} shows that the success probability converges to zero exponentially fast.

To understand why Chan et al's result \eqref{eq:chan} is not as strong as Theorem \ref{thm:bja}, 
we examine the proof in \cite{chan}.
At the heart of it lies an argument based on  Fano's inequality,  bounding the entropy $H( \vc{U} | \vc{Z}) $
using the decomposition
\begin{eqnarray}
H( \vc{U} | \vc{Z} = \vc{z})  =  H( E | \vc{Z} = \vc{z})  +  \pr( E = 1 | \vc{Z} = \vc{z}) H( \vc{U} | \vc{Z} = \vc{z}, E = 1),
\end{eqnarray}
where $E$ is the indicator of the error event $\vc{U} \neq \vc{Z}$.  In \cite{chan}
 this last term is bounded by $\log_2 \binom{N}{K}$, since a priori
$\vc{U}$ could be any defective set other than $\vc{z}$. However, in practice, this is 
a significant overestimate. For example, in the noiseless case there is a relatively small collection 
of defective sets  that a 
particular defective set $\vc{U}$ can mistakenly be estimated as
(referred to as $A(\cdot)$  later in this paper).
In this case, using inferences such as those  in the {\texttt{DD}} algorithm of \cite{johnson33}, 
any item which appears in a test pool $\sett{X}_t$ giving result $Y_t = 0$ cannot be defective. Essentially, Theorem
\ref{thm:bja} exploits this type of fact.

However  results corresponding to Theorem \ref{thm:bja} were not previously
known even for noiseless probabilistic group testing, let alone
more general settings, including noisy channels and other models for defectivity.
Tan and Atia \cite[Theorem 2]{tan} do prove a strong converse for combinatorial group testing, however, they do not achieve exponential decay. Since the result of the test only depends on whether the items in the defective set $\sett{K}$ are present,
we can restrict our attention to the submatrix $\mat{X}_{\sett{L}}$ indexed by subsets $\sett{L} \subseteq \sett{K}$.
\begin{theorem}[\cite{tan}, Theorem 2]
Write $\zeta_T := T^{-1/4} \pr(\suc)^{-1/2}$ and $\eta_T := T^{-1} + h(T^{-1/4})$. If the components of 
$\mat{X}$ are independent and identically distributed then for each 
$\sett{L}$, then $T$  (the number of tests required to achieve the given probability of success) satisfies:
$$ T( I( \mat{X}_{\sett{K} \setminus \sett{L}}; \mat{X}_{\sett{L}}, \vc{Y}) + \eta_T) \geq (1 - \zeta_T) \log_2 \binom{ N - |\sett{L}|}{K  - |\sett{L}|}.$$
\end{theorem}
Rearranging, and writing $I = I( \mat{X}_{\sett{K} \setminus \sett{L}}; \mat{X}_{\sett{L}}, \vc{Y})$  we obtain
$$ \pr(\suc) \leq \frac{1}{T^{1/2} \left( 1 - T (I + \eta_T)/ \log_2 \binom{ N - |\sett{L}|}{K  - |\sett{L}|} \right)^2}
\simeq \frac{1}{T^{1/2} \delta^2},$$
 taking $T^* =  \log_2 \binom{ N - |\sett{L}|}{K  - |\sett{L}|}/I$ for $T = (1-\delta) T^*$ . This gives a strong converse, though
not the exponential decay achieved in \eqref{eq:johnson} above. 
However Tan and Atia's results \cite{tan} are valid in a variety of settings and noise models.

Pedagogically, it is worth noting a parallel between these various approaches and treatments of the channel coding problem in the literature.
That is \eqref{eq:chan}, due to Chan et al. \cite{chan}, is proved using an argument based on Fano's inequality,
 parallelling the proof of Shannon's noisy coding theorem
 exemplified for example in \cite[Section 8.9]{cover}. The argument of  Tan and Atia \cite{tan} is based on 
Marton's blowing up lemma, 
mirroring the treatment of Shannon's theorem in the book of Csisz\'{a}r and K\"{o}rner \cite[Section 6]{csiszar6}. Our work in the 
non-adaptive case is based on the more  recent work of Polyanskiy, Poor and
Verd\'{u} \cite{polyanskiy2}, which has been adapted to the problem of data compression in \cite{kostina}.

The paper \cite{polyanskiy3}, written by the same authors as \cite{polyanskiy2}, extends their approach to channels with
feedback, which corresponds to the adaptive case of group testing. 
 We prefer to give bounds based on  older works of 
 Gallager \cite{gallager2}, Kemperman \cite{kemperman}, Kesten \cite{kesten} and Wolfowitz
\cite{wolfowitz3}.  Note that in the non-adaptive case, as described in \cite[Section III.G]{polyanskiy2}, the results of 
Polyanskiy et al. are stronger than the results of Wolfowitz \cite{wolfowitz3} and Gallager \cite{gallager2} type. However,
in the adaptive case, these earlier results appear easier to modify in the group testing context.
\section{Hypothesis testing and non-adaptive group testing} \label{sec:hyptest}
We first state a result, Theorem \ref{thm:ppvmain}, which implies strong converse results that generalize
Theorem \ref{thm:bja} in the non-adaptive case. The key observation comes from Polyanskiy, Poor and Verd\'{u} \cite{polyanskiy2} who found a 
relationship between channel coding  and hypothesis testing. Since the Neyman-Pearson lemma
gives us the optimal hypothesis test,  the paper \cite{polyanskiy2}  deduces strong bounds on coding error probabilities.
\begin{definition}
Write $\beta_{1-\epsilon}( P, Q)$ for the smallest possible type II error for hypothesis
tests (with type I error probability $\epsilon$)
 deciding between $P$ and $Q$. \end{definition}
Our contribution in this section is to use this same analogy for the group testing problem, given
a process generating random chosen defective sets $\vc{U}$ (a source).
To some extent this is simply a question of 
adapting the notation of \cite{polyanskiy2}. However, one generalization which is important for us is that we do not 
require that $\vc{U}$ is uniform (allowing us to consider probabilistic as well as combinatorial
group testing). This was not considered in \cite{polyanskiy2}, largely because for channel coding
it seems less natural to consider non-uniform $\vc{U}$.

Since we consider non-adaptive group testing, we fix $\mat{X} = \sett{X}$  in advance.
We write $\prv{KY}{k,y}$ for the joint probability distribution of $\vc{K}$ and $\vc{Y}$
and consider an algorithm which estimates (decodes)  the defective set $\vc{Z} = \whu$,
using only outputs $\vc{Y}$ and  test matrix $\mat{X} = \sett{X}$. 
Since $\mat{X} = \sett{X}$ is fixed, we simplify the notation of \eqref{eq:decoder} above and 
 write $\prv{Z|Y}{z|y}$ for the probability that the estimator gives $\vc{Z} = \vc{z}$ when $\vc{Y} = \vc{y}$.
We prove the following key result:
\begin{theorem} \label{thm:ppvmain} Suppose that the group testing model satisfies the Only Defects Matter property, Definition
\ref{def:odm}. For any non-adaptive choice of test design, any estimation rule $\prV{Z|Y}$ and probability mass function $\qrV{Y}$:
\begin{equation} \label{eq:bounds}
 \beta_{1-\epsilon}(\prV{KY}, \prV{K} \times \qrV{Y}) \leq \sum_{\vc{z} \in \{0,1 \}^N}  \prv{U}{z} \qst{z},
\end{equation} where  $\qst{z} = \sum_{\vc{y}} \qrv{Y}{y} \prv{Z|Y}{z|y}$ is the probability that $
\vc{Y} \sim \qrV{Y}$ is decoded to $\vc{z}$ and $1-\epsilon = \pr(\suc)$.
\end{theorem}
\begin{proof} See Appendix \ref{sec:proofppvmain}. \end{proof}
\begin{example} In the noiseless non-adaptive case consider any defective set distribution $\prV{U}$. 
Taking $\qrV{Y} \equiv 1/2^T$, the optimal rule is to accept $\prV{KY}$ with probability $1-\epsilon$ if
$\vc{x} = \vc{y}$, and to reject $\prV{KY}$ otherwise (this corresponds to taking $\lambda = 1-\epsilon$ and $d^* = 0$ in
Example \ref{ex:bsc} above). 
We obtain by Theorem \ref{thm:ppvmain} that 
\begin{equation} \label{eq:optimalch} (1-\epsilon)/2^T =  \beta_{1-\epsilon}( \prV{KY}, \prV{K} \times \qrV{Y})  \leq \sum_{\vc{z}  \in \{0,1 \}^N}  \prv{U}{z} \qst{z}.
  \end{equation}
\begin{enumerate}
\item {\bf [Uniform case]}
In particular, if  the $\prv{U}{z} = \II(\vc{z} \in \sett{M})/M$
for some set $\sett{M}$ of size $M = |\sett{M}|$, the RHS of \eqref{eq:optimalch}
becomes $ \leq 1/M$, this means that (exactly as in \cite[Theorem 27]{polyanskiy2}):
\begin{equation} \label{eq:uniform}
\beta_{1-\epsilon}( \prV{KY}, \prV{K} \times \qrV{Y}) \leq 1/M
\end{equation}
We deduce from Theorem \ref{thm:ppvmain} that 
\begin{equation} \label{eq:noiselessiid}
\pr(\suc) \leq 2^T/M, \end{equation}
confirming Theorem \ref{thm:bja} (under the additional assumption of non-adaptivitity; we discuss how to remove this
assumption in Theorem \ref{thm:noiseless} below).
\item \label{ex:noiseless} {\bf [General case]}
In general, in the noiseless non-adaptive case,
we write $\curP{m}$ for the sum of the largest $m$ values of $\prv{U}{z}$.
For each defective set $\vc{U}$, we  write $\vc{X} =  \vc{Y} = \vtheta(\vc{U})$. For a particular $\vc{y}$, we 
 write $A(\vc{y}) = \vtheta^{-1}(\vc{y}) = \{ \vc{z}: \vtheta(\vc{z}) = \vc{y} \}$ for the defective sets that get mapped to
$\vc{y}$ by the testing procedure.
We 
write $\pmax{y}  
= \max_{\vc{z} \in A(\vc{y})} \prv{U}{z}
$ for the maximum probability in
$A(\vc{y})$ and $\sett{U}^*(\vc{y}) = \{\vc{u}: \prv{U}{u} = \pmax{y} \}$ for the 
collection of defective sets achieving this probability.
For each $\vc{x}$, pick a string $\vc{u}^*(\vc{x}) \in \sett{U}^*(\vc{x})$ in any arbitrary fashion; and note that there are
up to $2^T$ strings $\vc{u}^*(\vc{x})$, which are distinct, since they each map to a different value under $\theta$.
These various definitions are illustrated in Figure \ref{fig:setfigure}.

In general  using  \eqref{eq:optimalch} we deduce that
 \begin{eqnarray} 
 \pr( \suc )  = (1-\epsilon) 
& \leq & 2^T \sum_{\vc{z} \in \{0,1 \}^N}  \prv{U}{z} \qst{z} \nonumber \\
& \leq &  \sum_{\vc{y} \in \{ 0, 1 \}^T }  \sum_{\vc{z} \in \{ 0, 1 \}^N} \prv{U}{z} \prv{Z|Y}{z|y}  \label{eq:matches} \\
& \leq & \sum_{\vc{y} \in \{ 0, 1 \}^T }  \sum_{\vc{z} \in A(\vc{y})} \pmax{y} \prv{Z|Y}{u|y} \label{eq:countinga} \\
& \leq &  \sum_{\vc{y} \in \{ 0, 1 \}^T } \pmax{y} \nonumber \\
& = & \sum_{\vc{y} \in \{ 0,1 \}^T} \prV{U} \left( \vc{u}^*(\vc{y}) \right) \nonumber \\
& \leq & \curP{2^T}. \label{eq:counting}
\end{eqnarray}
Here \eqref{eq:countinga} follows since for given $\vc{y}$
the success probability is maximised by restricting to $\prv{Z|Y}{z|y}$ supported on the
set  $\vc{z} \in A(\vc{y})$, so we know that $\prv{U}{z} \leq \pmax{y}$.
\eqref{eq:counting} follows since there are at most $2^T$ separate messages $\vc{X} = \vc{x}$, so at most $2^T$ 
distinct values $\vc{u}^*(\vc{x})$. This result generalizes \eqref{eq:noiselessiid}.
\end{enumerate}
\end{example}

Note that (as expected)
the success probability is maximised by the maximum likelihood decoder $\prV{Z|Y}$ which places all its support on
members of $\vc{U}^*(\vc{y})$.
In Theorem \ref{thm:noiseless} below we extend the result \eqref{eq:counting} to hold even in the adaptive case, extending Theorem \ref{thm:bja}.
Theorem \ref{thm:ppvmain} gives a converse for the non-adaptive binary symmetric channel case:
\begin{example} \label{ex:bsc}
Suppose
the output of standard combinatorial noiseless non-adaptive group testing 
$\vc{X}$ is fed through a memoryless binary symmetric channel with error probability $p < 1/2$ to produce $\vc{Y}$.
We write $x_i = \II (k_i \geq 1)$, and observe
that
$P(y_i | k_i) = (1-p)^{T-d(x_i,y_i)} p^{d(x_i,y_i)}$, where $d$ represents the Hamming distance.
Hence if $\qrV{Y} \equiv 1/2^T$, the likelihood ratio is
$$ \frac{ \prv{KY}{k,y}}{\prV{K} \times \qrv{Y}{k,y}} 
= \frac{ \prv{Y|K}{y|k}}{ \qrv{Y}{y}}
= \frac{  p^{d( \vc{x}, \vc{y})} (1-p)^{T - d( \vc{x}, \vc{y})}}
{1 /2^{T}} \propto \left( \frac{p}{1-p} \right)^{d( \vc{x}, \vc{y})}.$$
By the Neyman-Pearson lemma, the optimal rule is to accept $\prV{KY}$ if $d( \vc{x}, \vc{y}) < d^*$, to 
accept $\prV{KY}$ with probability $\lambda$ if $d( \vc{x}, \vc{y}) = d^*$ and to reject $\prV{KY}$ otherwise.
 In calculations which are essentially the same as in \cite[Theorem 35]{polyanskiy2}, we can find $d^*$ and $\lambda$
using \eqref{eq:uniform}:
\begin{eqnarray*}
\frac{1}{\binom{N}{K}} \geq
 \beta_{1-\epsilon}( \prV{KY}, \prV{K} \times \qrV{Y}) 
& = &  \pr( \bin(T,1/2) \leq d^*-1) + \lambda \pr (\bin(T, 1/2) = d^*).
\end{eqnarray*}
Then, for this value of $d^*$ we  write that 
\begin{eqnarray*}
 \pr(\suc) & = & 1- \pr(\mbox{\;type I error\;}) = \sum_{\vc{x}, \vc{y}} \prv{KY}{k,y} \pr(\mbox{\;accept $\prV{KY}$\;}) \\
& = &  \pr( \bin(T,p) \leq d^*-1) + \lambda \pr (\bin(T,p) = d^*).
\end{eqnarray*}
In Figure \ref{fig:bscnonadaptive}, 
we plot this in the case $N=500$, $K=10$, $p=0.11$, and for comparison plot the Fano bound taken from \cite[Theorem 2]{chan}:
\begin{equation} \label{eq:fanocombnoisy} \pr(\suc) \leq \frac{T(1-h(p))}{\log_2 \binom{N}{K}}. \end{equation}
In Figure \ref{fig:bscnonadaptiverates} we give the group testing analogue of \cite[Figure 1]{polyanskiy2}. We use the regime
of \cite{johnson33}; that is, we vary $N$ and take $K = \lceil N^{1-\beta} \rceil$, where $\beta = 0.37$ (this gives the value $K=10$ for $N=500$). Again taking $p=0.11$, we fix  $\pr(\suc) = 0.999$, and use
the lower bound on $T$ corresponding to the analysis above. This gives an upper bound on the
 rate $\log_2 \binom{N}{K}/T$, which we plot in Figure \ref{fig:bscnonadaptiverates}. Note that in this finite size regime, exactly
as in  \cite[Figure 1]{polyanskiy2}, the resulting rate bound is significantly smaller than the capacity $C= 1-h(p) = 0.500$, which
we only approach asymptotically.
\end{example}

\begin{figure}[ht!]
\begin{center}
\includegraphics[width=10cm]{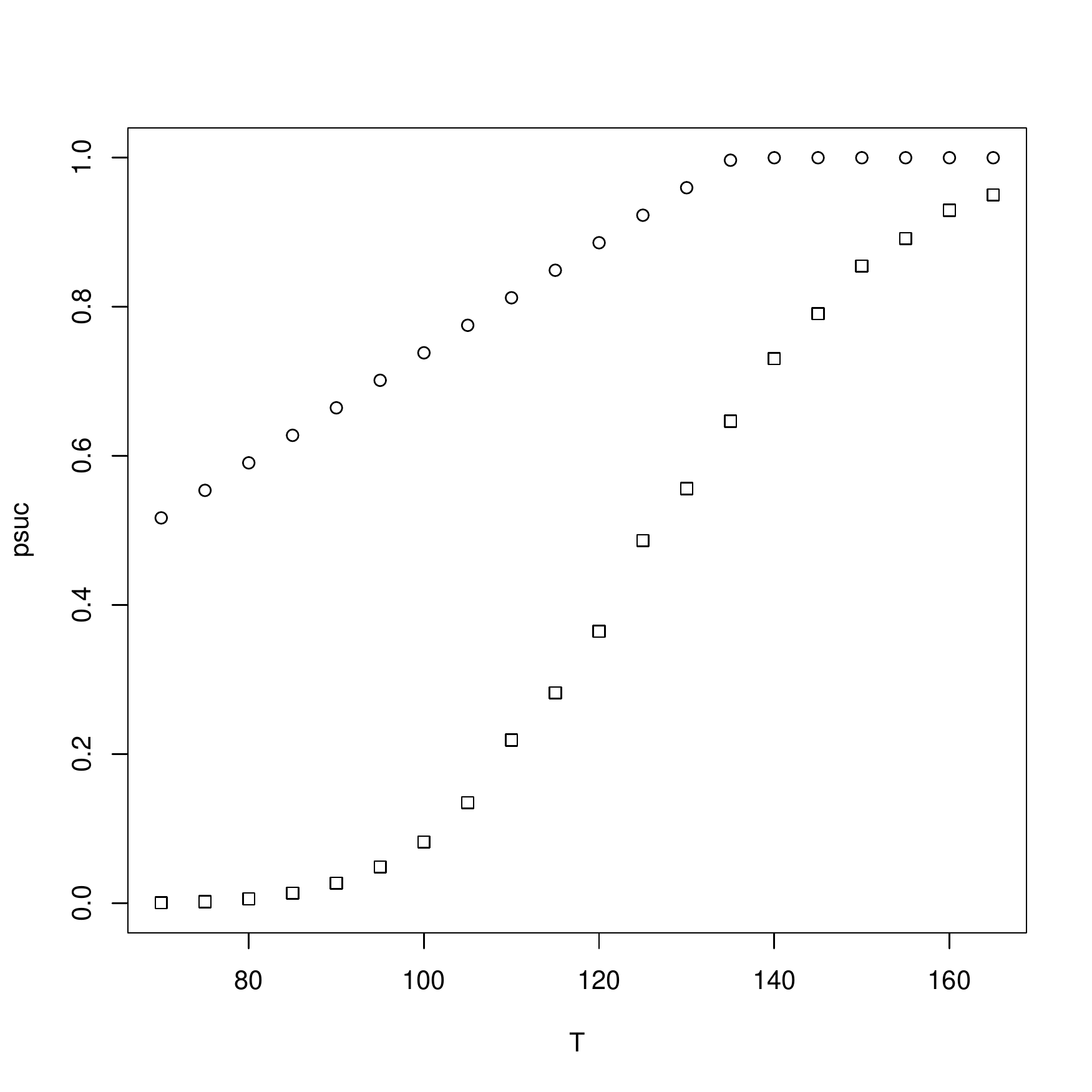}
\caption{Combinatorial non-adaptive 
group testing with $N = 500$ and $K=10$, where the output $\vc{X}$  of standard noiseless group testing  is fed into a
memoryless binary symmetric channel with $p=0.11$.
We vary the number of tests $T$ between 70 and 165, and plot the success
probability on the $y$ axis.
We plot the upper bound on $\pr(\suc)$ given by Example \ref{ex:bsc} using $\square$.
For comparison, we plot the (weaker) Fano bound \eqref{eq:fanocombnoisy}   taken
from \cite{chan} as $\circ$.
 \label{fig:bscnonadaptive}}
\end{center}
\end{figure}

\begin{figure}[ht!]
\begin{center}
\includegraphics[width=10cm]{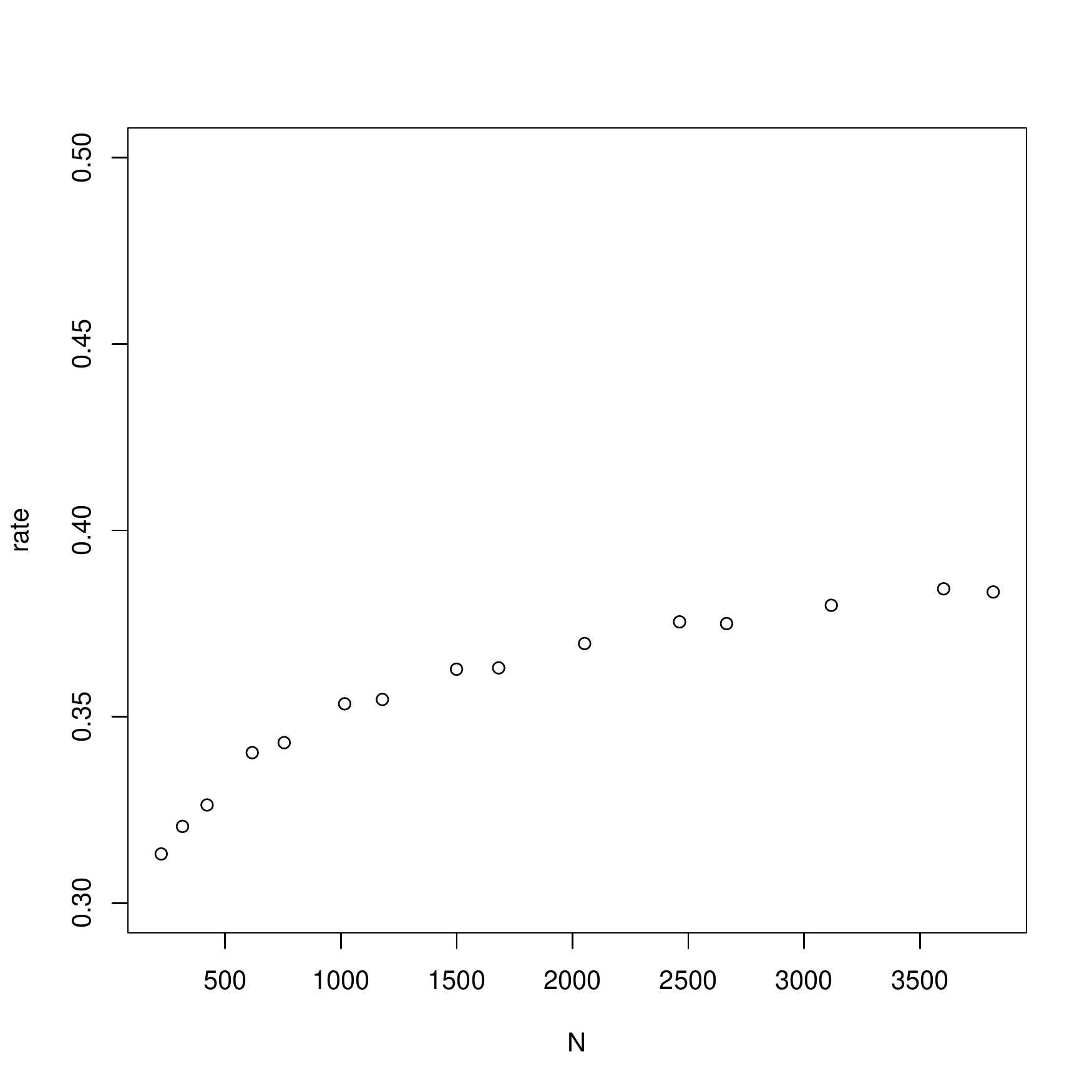}
\caption{Combinatorial non-adaptive 
group testing for various values of $N$ and $K=N^{0.37}$, where the output $\vc{X}$  of standard noiseless group testing  is fed into a
memoryless binary symmetric channel with $p=0.11$. In each case we choose $T$ large enough such that the 
success probability $\pr(\suc) = 0.999$.
We plot the upper bound on the rate given by Example \ref{ex:bsc}, and observe that this is significantly lower than
the value of the capacity $C = 0.500$ in this finite blocklength regime.
 \label{fig:bscnonadaptiverates}}
\end{center}
\end{figure}
\section{Adaptive group testing} \label{sec:adaptive}
We now consider  adaptive group testing, and
give a result (Proposition \ref{prop:basic}) which implies a strong converse, assuming that a concentration inequality is satisfied.
For any $t$, we write $\vc{Y}_{1,t} = \{ Y_1, \ldots, Y_t \}$ and $\sett{X}_{1,t} = \{ \sett{X}_1, \ldots, \sett{X}_t \}$.
We first prove the following representation result for the joint probability distribution of $(\vc{U}, \mat{X}, \vc{Y})$ under the model of adaptivity:
\begin{lemma} \label{lem:jointprob}
Assuming the Only Defects Matter property (Definition \ref{def:odm}) with transition matrix $\pr(Y_t = y|K_t = k) =
P(y|k)$ for all $k,y,t$, we can write
\begin{equation} \label{eq:jointprob}
 \prv{U, \mat{X}, Y}{u, \sett{X}, y} = \prv{U}{u}  \pcaus
\prod_{t=1}^T P( y_t | k_t), \end{equation}
where  $k_t = \vc{u}.\sett{X}_t$ is the number of defectives in the $t$th test and
\begin{equation} \label{eq:jointprob2}
\pcaus
:= \prod_{t=1}^T \prv{\mat{X}_t| \vc{Y}_{1,t-1}, \mat{X}_{1,t-1}}{ \sett{X}_t | \vc{y}_{1,t-1}, \sett{X}_{1,t-1} }
\end{equation}
is the causal conditional probability, with the key property that for any fixed $\vc{y}$:
\begin{equation} \label{eq:sumx} \sum_{\sett{X}} \pcaus = 1. \end{equation}
\end{lemma}
\begin{proof} We write (omitting the subscripts on $\pr$ for brevity) a collapsing product of the 
form:
\begin{eqnarray}
\pr( \vc{u}, \sett{X}, \vc{y}) & = & \pr( \vc{u}) \prod_{t=1}^T  \frac{ \pr \left( 
\vc{ u}, \sett{X}_{1,t}, \vc{y}_{1,t}  \right)}
{ \pr \left(  \vc{ u}, \sett{X}_{1,t-1}, \vc{y}_{1,t-1}  \right)} \nonumber \\
& = & \pr( \vc{u}) \prod_{t=1}^T  \pr \left(  \sett{X}_t, y_t |  \vc{ u}, \sett{X}_{1,t-1}, \vc{y}_{1,t-1}  \right)  \nonumber \\
& = & \pr( \vc{u}) \prod_{t=1}^T  \pr \left(  y_t |  \sett{X}_t, \vc{ u}, \sett{X}_{1,t-1},  \vc{y}_{1,t-1}  \right)  
 \pr \left(  \sett{X}_t |  \vc{ u}, \sett{X}_{1,t-1}, \vc{y}_{1,t-1}  \right) 
\label{eq:tosplit} \\
& = & \pr( \vc{u}) \prod_{t=1}^T  \pr \left(  y_t |  k_t \right)  
 \pr \left(  \sett{X}_t |   \sett{X}_{1,t-1}, \vc{y}_{1,t-1}  \right) \nonumber
\end{eqnarray}
where we remove the conditioning from the terms in \eqref{eq:tosplit} since $y_t$ is the result of sending $k_t = 
\vc{U}^T \sett{X}_t $ through a memoryless channel (the output of which is 
independent of previous test designs and their output) and since the choice of the $t$th test pool $\sett{X}_t$ is conditionally
independent of $\vc{U}$,  given the previous tests and their output.
\end{proof}
Next  we adapt arguments given by 
Wolfowitz \cite{wolfowitz3} which give strong converses for symmetric channels (even with feedback). Wolfowitz's book \cite{wolfowitz3} reviews earlier work of
his \cite{wolfowitz2}, and results of Kemperman \cite{kemperman} and Kesten \cite{kesten}.
  Wolfowitz \cite{wolfowitz3} and Kemperman \cite{kemperman} use Chebyshev's inequality to bound tail probabilities;
the fact that stronger results than Chebyshev could be used for the case without feedback was stated as \cite[Exercise 5.35]{gallager2}. 
Note the similarity of Proposition \ref{prop:basic} to \cite[Theorem 5.8.5]{gallager2}, a result originally due to Wolfowitz and
discussed for example as \cite[Theorem 10]{polyanskiy2}.
\begin{definition} \label{def:typset}
Fix a probability mass function $\qrV{Y}$ on $\{ 0, 1 \}^T$.
Define the  typical set $\typ$ by
\begin{equation} \label{eq:typset}
 \typ :=  \left\{   (\vc{k}, \vc{y}):     \frac{ \prod_{t=1}^T P \left( y_t | k_t \right)}{ \qrv{Y}{y} } \leq \exp( T  \beta)   \right\}
.\end{equation}
(Note that this set can be expressed in terms of the information density $ i(\vc{k} ; \vc{y})$ of \cite{polyanskiy2}).
\end{definition}
Write $\prv{Z|Y, \mat{X}}{z|y,\sett{X}}$ for the probability that some algorithm estimates the defective set as $\whu =
\vc{Z} = \vc{z} \in \{0,1\}^N$ when the group testing process with test matrix $\mat{X} = \sett{X}$
 returns $\vc{Y} =  \vc{y}$.
\begin{proposition} \label{prop:basic} Take any probability mass function $\qrV{Y}$ on $\{ 0, 1 \}^T$.
For any model of group testing (adaptive or non-adaptive), the success probability satisfies
\begin{equation} \label{eq:basic}
 \pr( \suc ) \leq  \exp( T \beta) \sum_{\vc{z} \in \{0,1 \}^N}   \prv{U}{z} \qst{z} +  \pr \left( (\vc{K}, \vc{Y}) \notin \typ \right) . 
\end{equation}
where we write $\qst{z} = \sum_{ \vc{y} , \sett{X} }  \prv{Z|Y, \mat{X}}{z|y,\sett{X}}
 \pcaus \qrv{Y}{y} $.
\end{proposition} 
\begin{proof} See Section \ref{sec:proofbasic}. \end{proof}
Note that by \eqref{eq:sumx}, we know that $\qst{z}$ is a probability mass function since:
\begin{eqnarray}
\sum_{\vc{z} \in \{ 0,1 \}^N} \qst{z} & = & \sum_{ \vc{y} , \sett{X} } 
 \pcaus \qrv{Y}{y}  \sum_{\vc{z} \in \{ 0,1 \}^N}  \prv{Z|Y, \mat{X}}{z|y,\sett{X}}  \nonumber \\
& = & \sum_{ \vc{y} \in \{0,1 \}^T} \qrv{Y}{y}
 \sum_{\sett{X} }  \pcaus  =   \sum_{ \vc{y} \in \{ 0,1 \}^T} \qrv{Y}{y} = 1.   \label{eq:qsumx} 
\end{eqnarray}
We use Proposition \ref{prop:basic} to 
 prove a result which extends Theorem \ref{thm:bja} for general defective set distributions  $\prV{U}$
in the noiseless binary case. This result applies to  both adaptive and non-adaptive group testing.
\begin{theorem} \label{thm:noiseless}
 For noiseless adaptive binary group testing,  if we write $\curP{m}$ for the sum of the largest $m$ values of $\prv{U}{z}$ then
$$ \pr(\suc) \leq \curP{2^T}.$$
\end{theorem}
\begin{proof} See Section \ref{sec:proofnoiseless}. \end{proof}
For combinatorial group testing, since $\prV{U}$ is uniform on a set of size $\binom{N}{K}$, Theorem \ref{thm:noiseless} implies that
$\curP{m} = m/\binom{N}{K}$ and we recover Theorem \ref{thm:bja}. We show how sharp this result is in Figure \ref{fig:bja}, which is reproduced from \cite[Figure 1]{johnsonc10}.
\begin{figure}[h]
\begin{center}
\includegraphics[width=12cm]{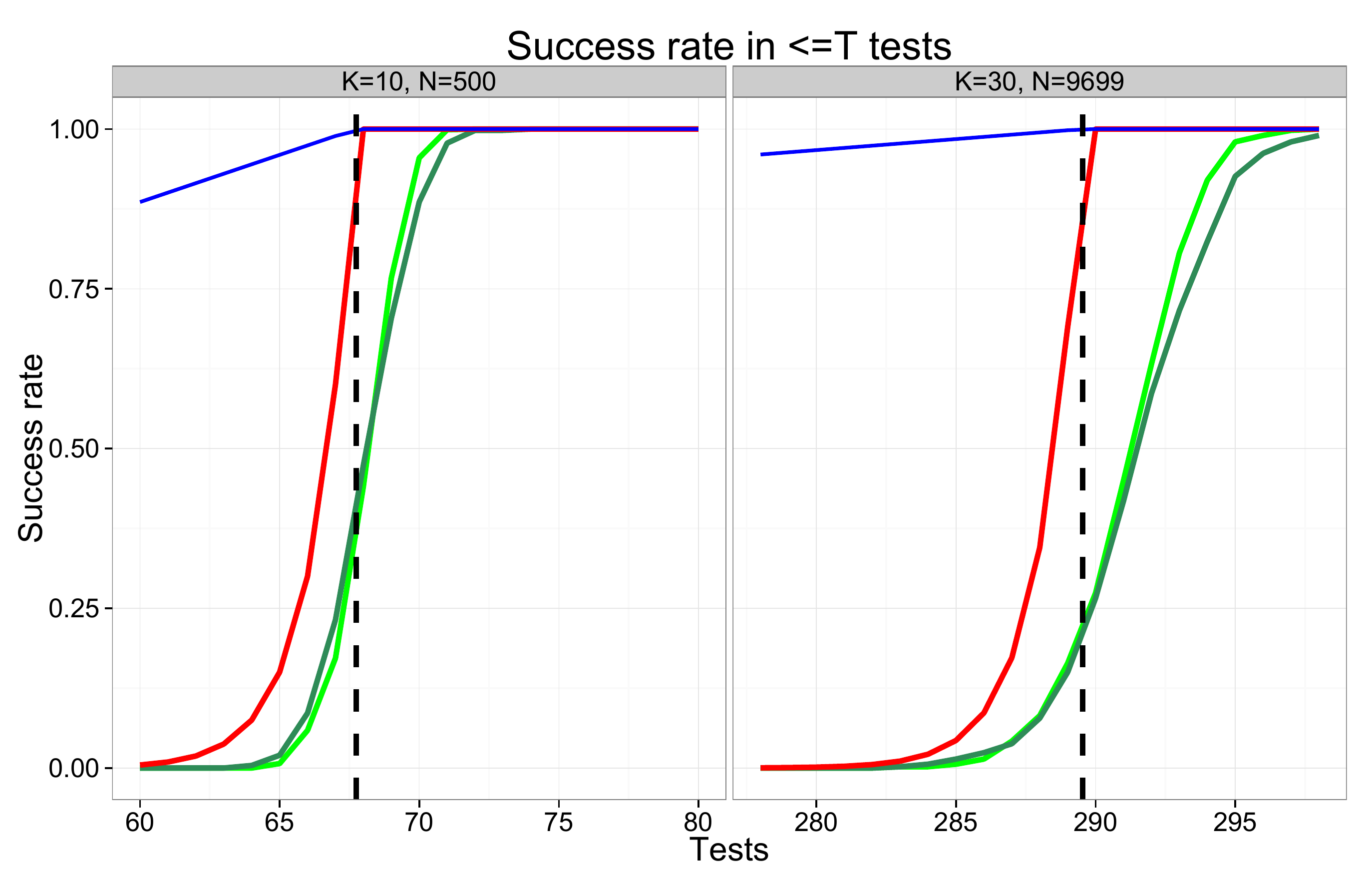}
\caption{(Reproduced from  \cite[Figure 1]{johnsonc10}).
Success probability for noiseless adaptive combinatorial group testing with
$(K,N) = (10, 500)$ and $(30, 9699)$ (these numbers are chosen to match the regime $K = N^{1-\beta}$, as in
Figure  \ref{fig:bscnonadaptiverates}).
 The upper bound on success probability
of Theorem \ref{thm:bja} is plotted in red, and the upper bound \eqref{eq:chan} (from  \cite[Equation (6)]{chan}) in
blue. The dotted vertical line is at $\log_2 \binom{N}{K}$ (the magic number).  To illustrate
how sharp Theorem \ref{thm:bja} is, we compare this with simulated (empirical) results of practical algorithms.
The empirical success probability of the
HGBSA of  Hwang \cite{hwang} is plotted as a bright green line, and the related algorithm analysed
in  \cite[Section IV]{johnsonc10} is plotted in dark green. 
 \label{fig:bja}}
\end{center}
\end{figure}
\begin{corollary} \label{cor:inforate}
Consider a sequence $\vc{U}^{(i)}$ of defectivity vectors of length $i$, generated as independent realisations of a stationary ergodic stochastic process of entropy rate $H$.
Given $T^{(i)} = (H-\epsilon) i$ tests to solve the $i$th noiseless adaptive group testing problem, the success probability tends to zero. (Hence the strong capacity cannot
be more than 1).
\end{corollary}
\begin{proof} 
We define the typical set
\begin{equation} \label{eq:typsmb}
\typsmb = \left\{ \left| \frac{ - \log \prv{U^{(i)}}{u}}{i} - H   \right| \leq \frac{\epsilon}{2} \right\} \end{equation}
By the Shannon-McMillan-Breiman theorem (AEP) (see for example \cite[Theorem 15.7.1]{cover}),
the probability $\pr(\typsmb) \rightarrow 1$.
Then, in Theorem \ref{thm:noiseless}, the $2^{T(i)}$ strings of largest probability will certainly be contained in a list containing the elements of 
$ (\typsmb)^c$  and the $2^{T(i)}$ strings
of largest probability in $\typsmb$. Since, by definition, any string in $\typsmb$ has probability less than $2^{- i H + i \epsilon/2}$, we deduce that
\begin{eqnarray*}
\pr(\suc) \leq  \curPi{2^{T(i)}} &  \leq & \pr \left( (\typsmb)^c \right) + 2^{T(i)} 2^{- i H + i \epsilon/2} \\
& = & \pr \left( (\typsmb)^c \right) +  2^{-  i \epsilon/2} 
\end{eqnarray*}
Given a quantitative form of the Shannon-McMillan-Breiman theorem (proved for example using the concentration inequalities described in \cite{raginsky}), we can deduce an explicit
(exponential) rate of convergence to zero of $\pr(\suc)$.
\end{proof}
We give more explicit bounds which show how  Theorem \ref{thm:noiseless}
can be applied in the noiseless probabilistic case in Section \ref{sec:noiselessexamples} below.
Section \ref{sec:adaptivenoisy} contains an illustrative
 example of  results that can be proved using Proposition \ref{prop:basic},
in the noisy adaptive case where $\vc{U}$ is uniformly distributed on a set $\sett{M}$ of size $M$ and  there
is a binary symmetric channel with error probability $p$ between $\vc{X}$ and $\vc{Y}$.
 Baldassini's thesis \cite{baldassini} developed and analysed algorithms 
in the noisy adaptive case. However, it remains an open problem to find capacity-achieving algorithms,
 even for examples such as the binary symmetric channel.
\section{Noiseless adaptive probabilistic group testing} \label{sec:noiselessexamples}
In this section, we give examples of bounds which can be proved using Theorem \ref{thm:noiseless} for noiseless adaptive
probabilistic
group testing. Note the similarity between the calculations in Examples \ref{ex:bsc} and \ref{ex:idprob}; in the former case
we control concentration of channel probabilities, in the latter we control source probabilities (see also  \cite[Theorem 35]{polyanskiy2}). Note that the control of the source strings with highest probabilities is an operation that lies at the analysis
of the finite blocklength data compression problem in \cite{kostina}.
\subsection{Noiseless probabilistic group testing with identical $p$}
\begin{example} \label{ex:idprob}
We consider the identical Probabilistic case, where $p_i \equiv p < 1/2$, so $\prv{U}{z} = p^w (1-p)^{N-w}$, where $w = w(\vc{z})$ is the Hamming
weight of $\vc{z}$. Write
\begin{equation} \label{eq:elnt}
 \elnt := \min \left\{ L: \sum_{i=0}^{L} \binom{N}{i} \geq 2^T \right\}
\end{equation}
and define $s \geq 0$ via
\begin{equation} \label{eq:ldef} 2^T = \sum_{i=0}^{\elnt-1} \binom{N}{i} + s, \end{equation}
meaning the $2^T$ highest probability defective sets are all of those of weight $\leq \elnt -1$, plus $s$ of
weight $\elnt$. We evaluate $\curP{2^T}$ in this case to obtain a bound which we plot in Figure \ref{fig:bernoulli}:
\begin{equation} \label{eq:curpval}
\pr(\suc) \leq
 \curP{2^T} = \sum_{i=0}^{\elnt-1} \binom{N}{i} p^{i} (1-p)^{N-i} + s  p^{\elnt} (1-p)^{N-\elnt}.\end{equation}
\end{example}

\begin{figure}[h]
\begin{center}
\includegraphics[width=12cm]{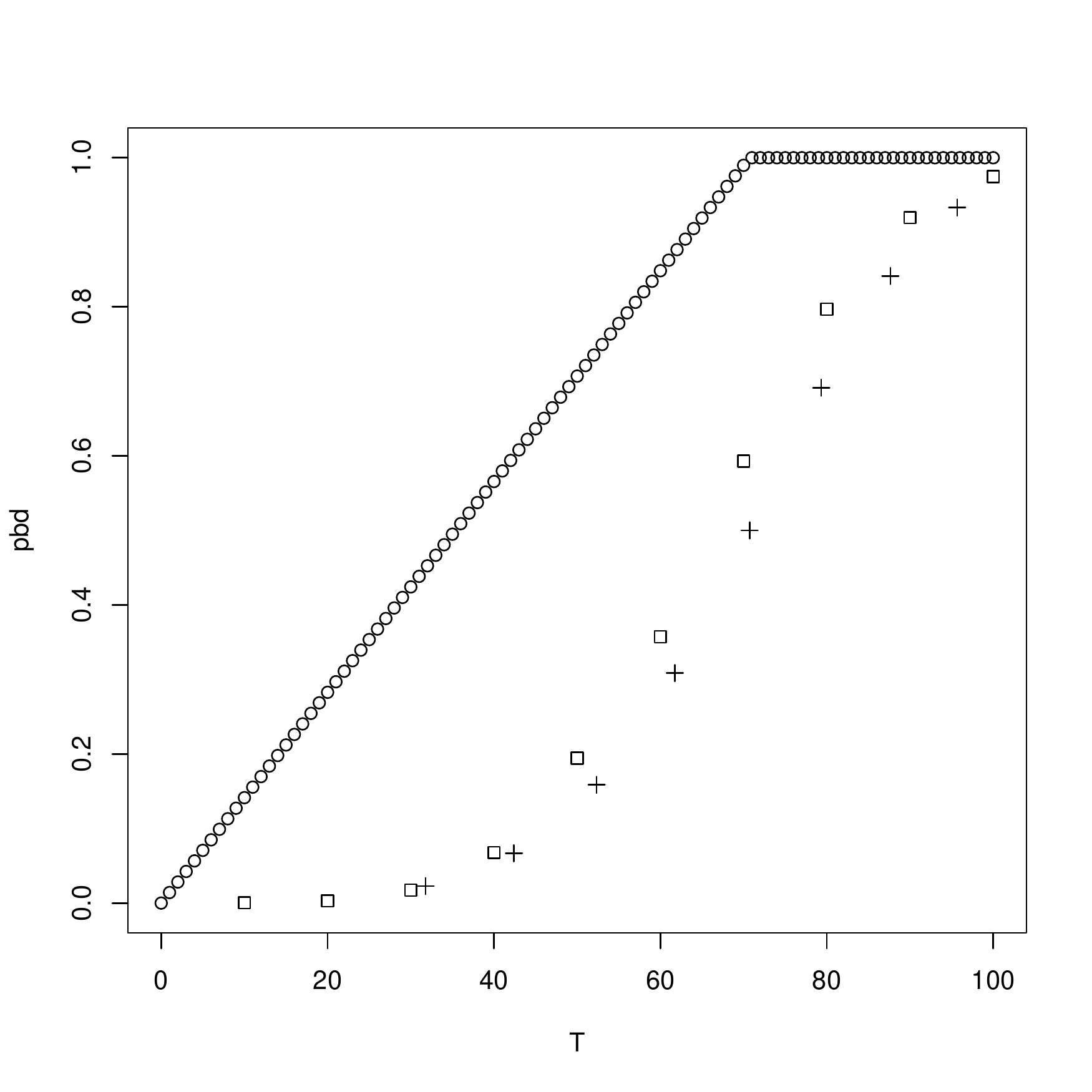}
\caption{Noiseless probabilistic adaptive group
testing in the case of $p_i \equiv 1/50$, $N = 500$. We vary the number of tests $T$ between 0 and 100, and plot the success
probability on the $y$ axis.
We plot the upper bound on $\pr(\suc)$ given by \eqref{eq:curpval} using $\square$.
For comparison, we plot the (weaker) Fano bound \eqref{eq:li} of \cite{li5} as $\circ$.
The approximation \eqref{eq:approxnorm} is plotted as $+$. \label{fig:bernoulli}}
\end{center}
\end{figure}

\begin{remark} We give a Gaussian approximation to the bound \eqref{eq:curpval}, in the spirit of \cite{polyanskiy2}.
Since we need to control tail probabilities we use the approximation given by 
Chernoff bounds (see Theorem \ref{thm:chernoff}).
In particular, if we take $L = L(y): = N p + y \sqrt{N p(1-p)}$ and $T(y)= N h(L(y)/N)$ then \eqref{eq:chernoff2} gives that
\begin{equation} 
\pr( \bin(N,1/2) \leq L(y)) \simeq 2^{ - N +T(y)},
\end{equation}
giving an approximate solution to \eqref{eq:elnt} as required.
Substituting  in \eqref{eq:curpval} we obtain
\begin{eqnarray}  \label{eq:approxnorm}
\pr(\suc) \leq \curP{2^{T(y)}}  & \simeq &  \pr( \bin(N, p) \leq L(y)) \simeq \Phi(y),
\end{eqnarray}
using a second normal approximation.
 For example, if  $y=0$ then $T = T(0) = N h(p)$ (the magic number) and
$L = N p$, and $\curP{2^T} \simeq \Phi(0) = 1/2$.
\end{remark}
Indeed using the Chernoff bound, we  use \eqref{eq:approxnorm} to deduce a strong capacity result:
\begin{corollary} Noiseless binary probabilistic group testing has strong
capacity $C=1$ in any regime where $p \rightarrow 0$ and $N p \rightarrow \infty$. \end{corollary}
\begin{proof}[Sketch proof]
For any $p \leq 1/2$ and $\epsilon > 0$, we  consider the asymptotic regime where $T = N h(p-\epsilon)$ as $N \rightarrow \infty$.
Choosing $L = N(p-\epsilon/2)$, we know that using standard bounds (see for example \cite[Equation (12.40)]{cover})
$$ \sum_{i=0}^L \binom{N}{i}  \geq  \binom{N}{L} \geq \frac{2^{Nh(L/N)}}{N+1} \geq \frac{2^{N h(p-\epsilon/2)}}{N+1},$$
which is larger than $2^T$ in the asymptotic regime.
Hence, summing over the strings of weight $\leq L$ will give at least the $2^T$ strings of highest probability, and we deduce 
by Theorem \ref{thm:chernoff} that
$$ \pr(\suc)  \leq \pr( \bin(N,p) \leq N(p - \epsilon/2)) \leq 2^{- N D(p -\epsilon/2 \| p)},$$
which tends to zero exponentially fast.
This complements the performance guarantee proved  in \cite{johnsonc12}, 
strengthening the result of \cite[Corollary 1.5]{johnsonc12} where the corresponding weak capacity result was stated 
using \eqref{eq:li}.
\end{proof}
\subsection{Noiseless probabilistic group testing with non-identical $p$}
In the case of probabilistic group testing, for non-identical $p_i$, the analysis is more complicated, and the form of 
the tightest bounds depends on the distribution of values of $p_i$. 
We assume that 
$1/2 \geq p_1 \geq p_2 \geq \ldots \geq p_N$ and write 
$\zeta_i = \log( (1-p_i)/p_i) \geq 0$. 
For a given value of $c$, we write $\sett{U}_c$ for the collection of defective sets with probability $\geq c$. That is
\begin{eqnarray}
\sett{U}_c & = & \{ \vc{u}: \pr(\vc{U} = \vc{u}) \geq c \} 
= \left\{ \vc{u}:  \sum_{i=1}^N (p_i - u_i) \zeta_i \geq \log c +  h(U) \right\}.
\end{eqnarray}
The key idea is to find a value $c$, 
such that we can guarantee that  $|\sett{U}_c| \geq 2^T$  (we discuss how to do this in Lemma \ref{lem:findc} below). Then, we  use concentration inequalities
to bound the total probability $\pr(|\sett{U}_c|)$ from above (this is done in Lemma \ref{lem:tailbound}). Then by construction  we know that 
$\curP{2^T} \leq \pr(|\sett{U}_c|)$, and we deduce an upper bound on the success probability, stated in Theorem \ref{thm:final}.
The details are given in Section \ref{sec:prooffinal}.
\begin{theorem} \label{thm:final}
Non-identical Probabilistic group testing has  success probability bounded by 
$$ \pr(\suc) \leq  \exp \left( -  \frac{(\log c^* + h( \vc{U})^2}{4 L} \right),  \mbox{ if $0 \leq (\log c^* + h( \vc{U})) \leq 
\frac{L}{M}$,} $$
where
\begin{equation} \label{eq:landm}
L = \sum_{i=1}^N p_i (1-p_i) \left( \log \left( \frac{1-p_i}{p_i} \right) \right)^2,
\mbox{\;\; and \;\;} M = \max_i  (1-p_i) \log \left( \frac{1-p_i}{p_i} \right), \end{equation}
 and we write
$$ c^*  = \max_{T \leq R \leq N} \left( \prod_{i=1}^{R-\elrt} (1-p_i) \prod_{i=R-\elrt+1}^R p_i \right). $$
\end{theorem}
\begin{proof} See Section \ref{sec:prooffinal}. \end{proof}

\section{Noisy adaptive group testing example} \label{sec:adaptivenoisy}
We now use Proposition \ref{prop:basic} to  prove a bound on $\pr(\suc)$  in a  noisy example.
For simplicity we state the following example in the case of uniform $\vc{U}$.
Further generalizations (in the spirit of Theorem \ref{thm:noiseless}) are possible by adapting the proofs along the lines
of Section \ref{sec:proofnoiseless}. 
\begin{example} \label{ex:adaptivenoisy}
Suppose $\vc{U}$ is uniformly distributed on a set $\sett{M}$ of size $M$.
and the noise channel between $\vc{X}$ and $\vc{Y}$ forms a binary symmetric channel.
Recall that \eqref{eq:basic} states that
\begin{equation} \label{eq:basic3}
 \pr( \suc ) \leq  \exp( T \beta) \sum_{\vc{z} \in \{0,1 \}^N}   \prv{U}{z} \qst{z} +  \pr \left( (\vc{K}, \vc{Y}) \notin \typ \right). 
\end{equation}
We write $C := 1 - h(p) 
$ for the Shannon capacity of the binary symmetric channel, and 
take
$\beta = C + \epsilon/2$.  By \eqref{eq:qsumx}, since $\prv{U}{z} \equiv 1/M$  the first term of
\eqref{eq:basic3} becomes $\exp(T (C + \epsilon/2))/M$.
Since $X_t = \II(K_t \geq 1)$, and the channel matrix $P(y  | k)$ is constant for values of $k \geq 1$, it collapses down
to a channel matrix from $x$ to $y$.
Further, taking $\qrV{Y} \equiv 1/2^T$, for $p < 1/2$,  the event that
\begin{eqnarray}
  \left\{ (\vc{k}, \vc{y}) \notin \typv{C+\epsilon/2} \right\}  & \Longleftrightarrow &  \left\{
\frac{1}{T}  \left( \sum_{i=1}^T \log P(y_i | k_i) - \log Q(y_i) \right) >  C + \frac{\epsilon}{2} \right\}  \nonumber \\
& \Longleftrightarrow &  \left\{ 1 + \frac{d(\vc{x}, \vc{y})}{T} \log p + \left( 1 -  \frac{d(\vc{x}, \vc{y})}{T} \right)
\log(1-p)  >  C + \frac{\epsilon}{2} \right\} \nonumber \\
& \Longleftrightarrow &  \left\{ - 
\left( d(\vc{x}, \vc{y}) - T p \right)  \log \left( \frac{(1-p)}{p} \right) >  
\frac{\epsilon T}{2} \right\} \label{eq:asymp}
\end{eqnarray}
We  deduce results, both in the asymptotic (capacity) sense and the finite blocklength regime.
\begin{enumerate}
\item
Consider a sequence of group testing problems, where the $i$th problem has
$\vc{U}^{(i)}$  uniformly distributed on a set $\sett{M}(i)$ of size $M(i)$, we can deduce that any sequence of algorithms using
$T(i)$ tests has
\begin{equation} \label{eq:lowerval}
        \liminf_{i \tends} \frac{ \log M(i)}{ T(i)} = \frac{ H(\vc{U}^{(i)}) }{T(i)} \geq C+ \epsilon,
      \end{equation}
      has success probability $\pr(\suc) \rightarrow 0$, and hence strong capacity bounded above by $C$. 
This follows by the considerations above, since \eqref{eq:asymp} tells us that \eqref{eq:basic3} becomes
\begin{eqnarray*} \pr(\suc) & \leq & \frac{1}{M} \exp(T (C + \epsilon/2)) + \pr( \bin(T,p) \leq T q) \\
& \leq & \exp(-T \epsilon/2) + 2^{ - T D(q \| p)}
\end{eqnarray*}
where $q = p - \epsilon/(2 \log((1-p)/p))$,
and we deduce (exponential) convergence to zero, using the Chernoff bound Theorem \ref{thm:chernoff} below.

\item For any $d^*$, we can consider the set $ \{ d \leq d^* \}$ which (by \eqref{eq:asymp})
 corresponds to taking $\epsilon T/2 =
- (d^* - T p) \log((1-p)/p)$.
Then
\eqref{eq:basic3} becomes
\begin{eqnarray}
\pr(\suc) & \leq & \min_{d^*} \left[ \frac{1}{M} \exp \left( T C + (d^* - T p) \log \left( \frac{(1-p)}{p} \right) \right) + \pr \left( \bin(T,p) \leq d^*  \right) \right]. \;\;\;\; 
\label{eq:toplotbscadap}
\end{eqnarray}
We illustrate this bound in  Figure \ref{fig:bscadaptive}, where we compare it with the bounds derived in the adaptive case in
Example \ref{ex:bsc}.
\end{enumerate}
\end{example}

\begin{figure}[ht!]
\begin{center}
\includegraphics[width=10cm]{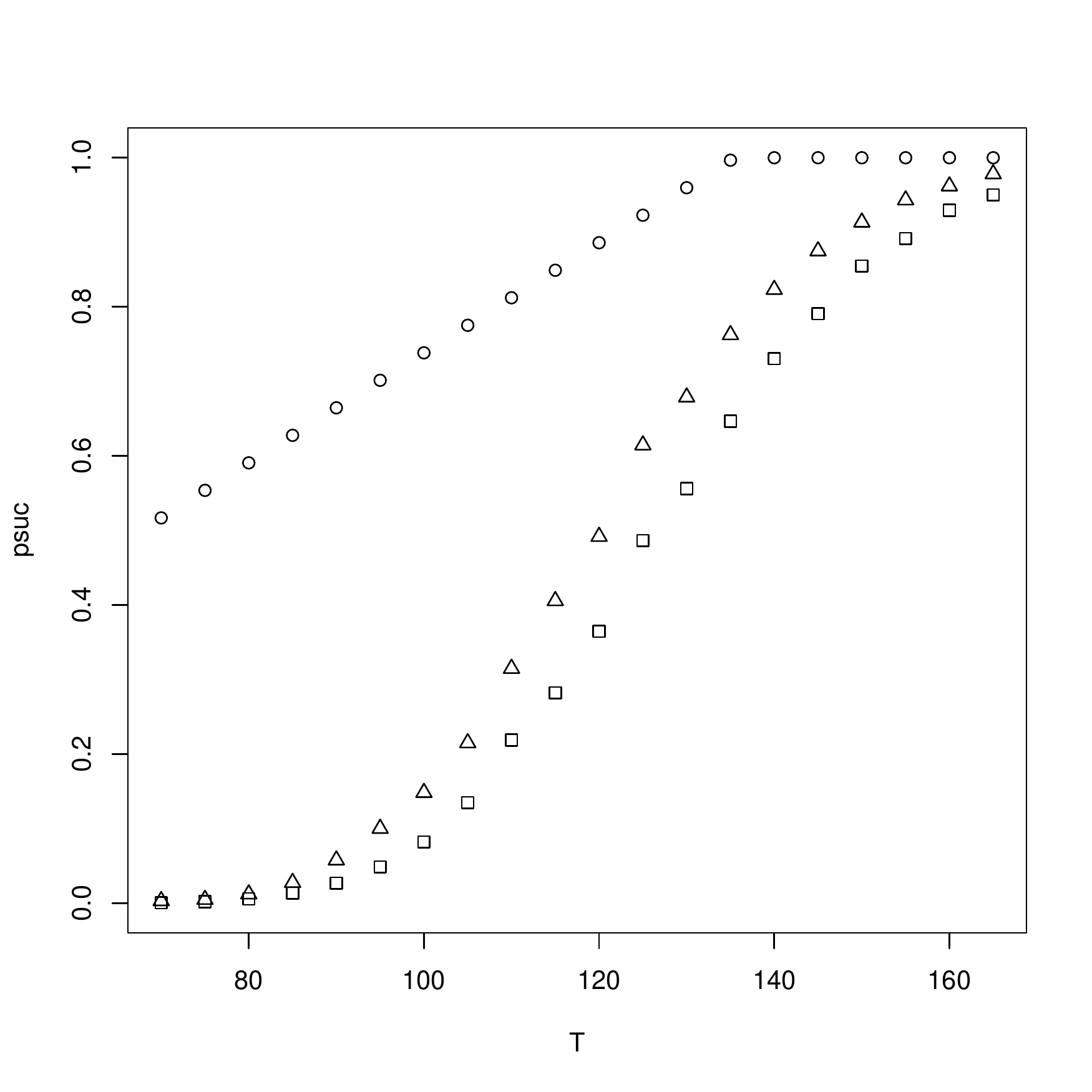}
\caption{Combinatorial non-adaptive 
group testing with $N = 500$ and $K=10$, where the output $\vc{X}$  of standard noiseless group testing  is fed into a
memoryless binary symmetric channel with $p=0.11$.
We vary the number of tests $T$ between 70 and 165, and plot the success
probability on the $y$ axis. 
We plot the bound \eqref{eq:toplotbscadap} (from Example  \ref{ex:adaptivenoisy}) on the success probability for adaptive algorithms 
  as $\triangle$. Since this is exactly the same scenario as Figure \ref{fig:bscnonadaptive}, we 
add the points from that figure for comparison.
That is, we plot the upper bound on $\pr(\suc)$ for non-adaptive algorithms given by Example \ref{ex:bsc} using $\square$,
showing a small adaptivity gap between upper bounds.
Again, we plot the (weaker) Fano bound \eqref{eq:fanocombnoisy}   taken
from \cite{chan} as $\circ$.
 \label{fig:bscadaptive}}
\end{center}
\end{figure}

\clearpage
 
\appendix

\section{Proof of Theorem \ref{thm:ppvmain}} \label{sec:proofppvmain}
We use an argument based on \cite{polyanskiy2}, adapted to the scenario where $\vc{U}$ need not be uniform. Consider a  hypothesis testing problem where we are given a pair $(\vc{k},\vc{y})$ and asked to test the null hypothesis that  it comes from joint distribution $\prV{KY}$ against an alternative of some other specific
 $\qrV{KY}$. This is a counterfactual exercise; in group testing 
we do not know $\vc{K}$, however, it is helpful to imagine a separate user who is  asked
to make inference using this information, and uses
 the following hypothesis testing rule:
\begin{quote}
given pair $(\vc{k},\vc{y})$ send $\vc{y}$ to the decoder to produce $\vc{z}$, and then accept $\prV{KY}$ with probability
\begin{equation} \label{eq:selection} 
 \prv{U|K}{z|k}  = \frac{ \prv{U}{z} \prv{K|U}{k|z} }{\prv{K}{k}}. \end{equation}
\end{quote}
\begin{proof}[Proof of Theorem \ref{thm:ppvmain}]
The key is to notice that $\vc{U} \rightarrow \vc{K} \rightarrow \vc{Y} \rightarrow \vc{Z}$ form a Markov chain, so
for estimation algorithm $\prV{Z|Y}$ we obtain 
\begin{eqnarray*}
\prv{Z|U}{w|z} 
& = & \sum_{\vc{k},\vc{y}} \prv{Z|Y}{w|y} \prv{Y|K}{y|k} \prv{K|U}{k|z}.
\end{eqnarray*}
Using this, there is an equivalence betwen error probability and 
$\pr(\mbox{Type I error})$  since
\begin{eqnarray}
\pr(\suc) & = & \sum_{\vc{z},\vc{w}}  \prv{U}{z}  \prv{Z|U}{w|z} \II(  \vc{w} = \vc{z}) \nonumber \\
& = & \sum_{\vc{z},\vc{w}}  \prv{U}{z}  \left[ \sum_{\vc{k},\vc{y}} \prv{Z|Y}{w|y} \prv{Y|K}{y|k} \prv{K|U}{k|z}
\right] \II(  \vc{w} = \vc{z}) \nonumber \\
& = & \sum_{\vc{k},\vc{y}}  \prv{K}{k}   \prv{Y|K}{y|k} 
\sum_{\vc{z},\vc{w}} \II(  \vc{w} = \vc{z})   \prv{Z|Y}{w|y} \frac{ \prv{U}{z} \prv{K|U}{k|z}}{\prv{K}{k}} 
\nonumber \\
& = & \sum_{\vc{k},\vc{y}}  \prv{KY}{k,y}   
\sum_{\vc{z}}   \prv{Z|Y}{z|y}
\frac{ \prv{U}{z} \prv{K|U}{k|z}}{\prv{K}{k}} \label{eq:todeal} \\
& = & \sum_{\vc{k},\vc{y}}  \prv{KY}{k,y}   \sum_{\vc{z}} \prv{Z|Y}{z|y}  \prv{U|K}{z|k}  \nonumber \\
& = & \sum_{\vc{k},\vc{y}}  \prv{KY}{k,y}  \pr(\mbox{accept $\prV{KY}$ given pair $(\vc{k}, \vc{y})$}) \nonumber \\
& = & 1 - \pr( \mbox{\;Type I error\;}) \nonumber
\end{eqnarray}
where we use the expression \eqref{eq:selection} to deal with \eqref{eq:todeal}.
We  find the probability of a Type II error in the same way. We focus on the case where $\qrV{KY} = \prV{K} \times
 \qrV{Y}$ (so $\vc{K}$ and $\vc{Y}$  are independent under $\qrV{KY}$), where
\begin{eqnarray}
\pr(\mbox{\;Type II error\;})
& = & \sum_{\vc{k},\vc{y}}  \qrv{KY}{k,y}  \pr( \mbox{accept $\prV{KY}$ given pair $(\vc{k}, \vc{y})$}) \nonumber \\
& = & \sum_{\vc{k},\vc{y}}  \prv{K}{k} \qrv{Y}{y} 
\sum_{\vc{z}}   \prv{Z|Y}{z|y}
\frac{ \prv{U}{z} \prv{K|U}{k|z}}{\prv{K}{k}} \nonumber \\
& = & \sum_{\vc{k}, \vc{z}}  \prv{U}{z} \prv{K|U}{k|z}  \sum_{\vc{y}} \qrv{Y}{y}  \prv{Z|Y}{z|y}  \nonumber \\
& = & \sum_{\vc{k}, \vc{z}}  \prv{U}{z} \prv{K|U}{k|z}  \qst{z}  \nonumber \\
& = & \sum_{\vc{z}}  \prv{U}{z} \qst{z}  \label{eq:tocontrol}
\end{eqnarray}
Hence, since $\beta_{1-\epsilon}$ gives the minimum type II error, we deduce the Proposition.
\end{proof}
\section{Proof of Proposition \ref{prop:basic}} \label{sec:proofbasic}
\begin{proof}[Proof of Proposition \ref{prop:basic}]
We write $\intyp$ for the event $\left\{ ( \vc{K}, \vc{Y} ) \in \typ \right\}$ and consider
\begin{eqnarray}
\pr(\suc) & =  &  \pr\left(\suc \bigcap \intyp \right) +  \pr\left(\suc \bigcap \intyp^c \right) \nonumber \\
& \leq & \pr\left(\suc \bigcap \intyp \right) +  \pr\left( \intyp^c \right). \label{eq:innersum} 
\end{eqnarray}
We write $\vc{u}$, $\sett{X}$ and $\vc{y}$ as indices of summation for brevity, to refer to sums over
$\vc{u} \in \{ 0,1 \}^N$, $\sett{X} \in \{0,1 \}^{N \times T}$ and $\vc{y} \in \{0,1 \}^T$.
Using the fact that  for the estimation algorithm $\prV{Z|Y,\mat{X}}$
$$\pr(\suc |  \vc{U} = \vc{u}, \mat{X} = \sett{X}, \vc{Y} = \vc{y}) 
= \sum_{\vc{z}} \prv{Z|Y,\mat{X}}{z|y, \sett{X}} \II ( \vc{z} = \vc{u}),$$ 
the first term of \eqref{eq:innersum} becomes:
\begin{eqnarray*}
\lefteqn{
\sum_{ (\vc{u}^T \sett{X}, \vc{y}) \in \typ } 
 \pr( \suc | \vc{U} = \vc{u} , \mat{X} = \sett{X}, \vc{Y} = \vc{y} ) \prv{U, \mat{X},Y}{u ,\sett{X},y}  }  \hspace*{2cm} \\
& \leq & \exp(T \beta)
\sum_{ (\vc{u}^T \sett{X}, \vc{y}) \in \typ } 
 \pr( \suc | \vc{U} = \vc{u} , \mat{X} = \sett{X}, \vc{Y} = \vc{y} )
\prv{U}{u} \pcaus \qrv{Y}{y} \\
& \leq & \exp(T \beta)
\sum_{ (\vc{u}, \sett{X}, \vc{y}) }  \left[  \sum_{\vc{z}} \prv{Z|Y,\mat{X}}{z|y, \sett{X}} \II ( \vc{z} = \vc{u}) \right]
\prv{U}{u} \pcaus \qrv{Y}{y} \\
& = & \exp(T \beta) \sum_{\vc{z}} \prv{U}{z}
\sum_{ \vc{y}, \sett{X} }   \prv{Z|Y,\mat{X}}{z|y, \sett{X}}  \pcaus
  \qrv{Y}{y}  \\
& = & \exp(T \beta) \sum_{\vc{z}}   \prv{U}{z} \qst{z} 
\end{eqnarray*}
where  we write $\qst{z} = \sum_{ \vc{y}, \sett{X} }  \prv{Z|Y,\mat{X}}{z|y, \sett{X}}  \pcaus
  \qrv{Y}{y}$.
\end{proof}

\section{Proof of Theorem \ref{thm:noiseless}} \label{sec:proofnoiseless}

\begin{proof}[Proof of Theorem \ref{thm:noiseless}]
In general, in the noiseless case, for each defective set $\vc{U}$, we  write $\vc{X} =  \vc{Y} = \vtheta(\vc{U},
\sett{X})$. For a particular $\vc{Y} = \vc{y}$ and $\mat{X} = \sett{X}$, we 
 write $A(\vc{y}, \sett{X}) = \vtheta^{-1}(\vc{y}, \sett{X}) = \{ \vc{z}: \vtheta(\vc{z}, \sett{X}) = \vc{y} \}$ for the defective sets that get mapped to
$\vc{y}$ by the testing procedure defined by $\sett{X}$.
We 
write $\pmax{y, \sett{X}} 
= \max_{\vc{z} \in A(\vc{y}, \sett{X})} \prv{U}{z}
$ for the maximum probability in
$A(\vc{y}, \sett{X})$ and $\sett{U}^*(\vc{y}, \sett{X}) = \{\vc{u}: \prv{U}{u} = \pmax{y, \sett{X}} \}$ for the 
collection of defective sets achieving this probability.
For each $\vc{y}$, pick a string $\vc{u}^*(\vc{y},\sett{X}) \in \sett{U}^*(\vc{y},\sett{X})$ in any arbitrary fashion; and note that there are
up to $2^T$ strings $\vc{u}^*(\vc{y}, \sett{X})$, which are distinct, since they each map to a different value under $\vtheta(\cdot, 
\sett{X})$.
These definitions are illustrated in Figure \ref{fig:setfigure}.

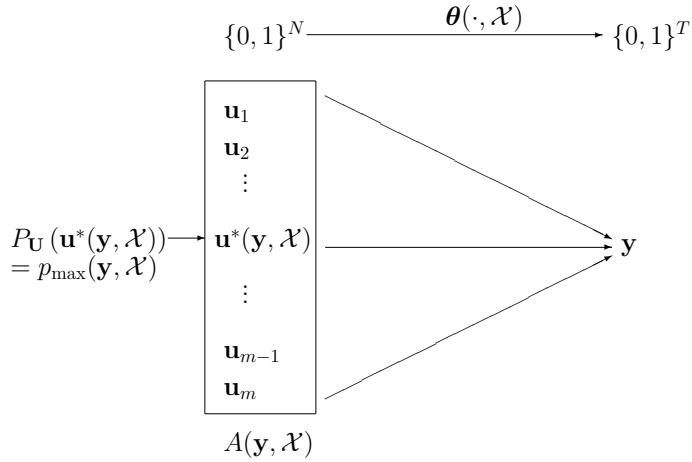
\begin{figure}[ht!]
\begin{large}
\begin{center}
\scalebox{0.7}{
\begin{picture}(400,250)(-20,5)
\put(100,10){$A(\vc{y},\sett{X})$}
\put(100,40){$\vc{u}_m$}
\put(100,60){$\vc{u}_{m-1}$}
\put(110,90){$\vdots$}
\put(95,120){$\vc{u}^*(\vc{y},\sett{X})$}
\put(110,150){$\vdots$}
\put(100,170){$\vc{u}_2$}
\put(100,190){$\vc{u}_1$}
\put(100,230){$\{0,1\}^N$}
\put(90,30){\line(1,0){60}}
\put(90,210){\line(1,0){60}}
\put(90,30){\line(0,1){180}}
\put(150,30){\line(0,1){180}}
\put(-15,120){$\prv{U}{\vc{u}^*(\vc{y},\sett{X})}$}
\put(-15,105){$ = \pmax{y,\sett{X}}$}
\put(70,125){\vector(1,0){20}}
\put(310,230){$\{0,1 \}^T$}
\put(145,235){\vector(1,0){160}}
\put(315,118){$\vc{y}$}
\put(155,120){\vector(1,0){155}}
\put(155,202){\vector(2,-1){155}}
\put(155,38){\vector(2,1){155}}
\put(220,240){$\vd{\theta}(\cdot, \sett{X})$}
\end{picture}}
\caption{Schematic illustration of the sets used in the proof of Theorem \ref{thm:noiseless} \label{fig:setfigure}}
\end{center}
\end{large}
\end{figure}
We use the bound from Proposition \ref{prop:basic},
taking $\qrv{Y}{y} \equiv 1/2^T$ and $\beta = \ln 2$ in \eqref{eq:typset}, so that $\typ = \left\{ (\vc{k}, \vc{y}):
 \prv{Y|K}{y|k} \leq 1 \right\}$, which holds automatically. That is, the set $\typ^c$ is the empty set and so
the second term in \eqref{eq:basic} vanishes. 
We analyse the first term in \eqref{eq:basic} repeating arguments from Example  \ref{ex:noiseless} to obtain:
\begin{eqnarray} 
 \pr( \suc ) & \leq & 2^T  \sum_{\vc{y}, \sett{X}} \pcaus \qrv{Y}{y}
\sum_{\vc{z}} \prv{U}{z}  \prv{Z|Y, \mat{X}}{z|y,\sett{X}}  \label{eq:matchesad} \\
& \leq &  2^T  \sum_{\vc{y}, \sett{X}} \pcaus \qrv{Y}{y}
 \sum_{\vc{z} \in A(\vc{y})} \pmax{y}  \prv{Z|Y, \mat{X}}{z|y,\sett{X}} \label{eq:countingaad} \\
& \leq &   2^T  \sum_{\vc{y},\sett{X}} \qrv{Y}{y} \pmax{y} \pcaus   \nonumber \\
& = &    \sum_{\vc{y}}  \pmax{y} \sum_{\sett{X}} \pcaus   \nonumber \\
& = & \sum_{\vc{y}} \prV{U} \left( \vc{u}^*(\vc{y}) \right) \label{eq:usesumx} \\
& \leq & \curP{2^T}. \label{eq:countingad}
\end{eqnarray}
Here \eqref{eq:countingaad} follows since for given $\vc{y}$
the success probability is maximised by restricting to $\prv{Z|Y,\mat{X}}{z|y,\sett{X}}$ supported on the
set  $\vc{z} \in A(\vc{y}, \sett{X})$, so we know that $\prv{U}{z} \leq \pmax{y}$.
\eqref{eq:usesumx} follows using  \eqref{eq:sumx}.
\eqref{eq:countingad} follows since there are at most $2^T$ separate messages $\vc{Y} = \vc{y}$, so at most $2^T$ 
distinct values $\vc{u}^*(\vc{y})$.
\end{proof}
\begin{remark} Note the striking fact that \eqref{eq:countingad} exactly matches \eqref{eq:matches}. That is, although they are 
proved by very different methods, our best results in the noiseless adaptive and non-adaptive cases coincide. This may suggest
that there is not an `adaptivity gap' (in the language of \cite{aldridge3,johnson33}) in this case, though of course stronger
converses may be possible. \end{remark}

\section{Proof of Theorem \ref{thm:final}} \label{sec:prooffinal}

\begin{lemma} \label{lem:findc} For each $R$ with $T \leq R \leq N$, we  find a set $\sett{U}_c$ with $|\sett{U}_c| \geq 2^T$, and
where 
\begin{equation} \label{eq:cvalue}
c = c(R) = \prod_{i=1}^{R-\elrt} (1-p_i) \prod_{i=R-\elrt+1}^R p_i. \end{equation}
We write $c^* = \max_{T \leq R \leq N} c(R)$.
\end{lemma}
\begin{proof}
Given $R$, we take $\sett{S}_R = \{1, 2, \ldots, R\}$ to be the set of items with the $R$ largest values of $p_i$, and form defective sets using $\sett{S}_R$ only.
Using Equation \eqref{eq:elnt}, we know that we can find at least $2^T$ defective sets by just using subsets of
$\sett{S}_R$ with  weight $\elrt$. 
The smallest probability of such a set is given in Equation \eqref{eq:cvalue}.
\end{proof}

\begin{remark} Note that 
for $R=T$, we need $\elnt = T$, and obtain $c(R) = \prod_{i=1}^T p_i$.
In the IID case, for $R=N$, we take sets of weight $\elnt$ and recover Example \ref{ex:idprob} above.
\end{remark}

\begin{lemma} \label{lem:tailbound}
Given $c$, we can bound the probability 
$$ \pr(|\sett{U}_c| ) \leq \exp \left( -  \frac{(\log c + h(U)^2}{4 L} \right),  \mbox{ if $0 \leq (\log c + h(U)) \leq L/M$,} $$
where  $L$ and $M$ are defined in \eqref{eq:landm}.
\end{lemma}
\begin{proof} We apply Bernstein's inequality (see for example \cite[Theorem 2.8]{petrov}) to the sum of 
zero mean random variables $(p_i - U_i) \zeta_i$.
\end{proof}
\section{Concentration inequality}
We require an exponential bound in terms of relative entropy. There is a wide literature on this subject, and we 
 take a one-sided form of the Chernoff bound stated as  \cite[Theorem 5]{raginsky} (for $p \leq 1/2$, we take $d = (1-p)$
and $\sigma^2 = p(1-p)$ in the result stated there):
\begin{theorem} \label{thm:chernoff}
For $q < p \leq 1/2$, we  bound the probability
$$ \pr \left( \bin(n,p) \leq nq  \right) \leq 2^{ - n D( q \| p)},$$
where we write $D(q \| p)$ for the relative entropy from a Bernoulli($q$) random variable to a Bernoulli($p$),
calculated using logarithms to base 2. \end{theorem}
Since this is generally a tight bound, we use it to motivate the following approximation, which comes from writing $D( q \| 1/2)
= \log 2 - h(q)$. For any $L$ we deduce that
\begin{equation} \label{eq:chernoff}
\pr( \bin(N,1/2) \leq L) \simeq 2^{- N D(L/N \| 1/2) } = 2^{-N} 2^{  N h(L/N) }.
\end{equation}
If we take $L = L(y): = N p + y \sqrt{N p(1-p)}$ and $T(y)= N h(L(y)/N)$  we deduce that
\begin{equation} \label{eq:chernoff2}
\pr( \bin(N,1/2) \leq L(y)) \simeq  2^{  - N +T(y)}.
\end{equation}
\section*{Acknowledgements}
The author thanks Matthew Aldridge, Leonardo Baldassini and Thomas Kealy for useful discussions regarding the
group testing problem, and Vanessa Didelez for help in understanding causal conditional probability.
%
%

\end{document}